\documentclass[11pt,a4paper]{article}
\usepackage[T1]{fontenc}
\usepackage[latin1]{inputenc}
\usepackage[english]{babel}
\usepackage{amsmath}
\usepackage{amsthm}
\usepackage{amsfonts}
\usepackage{amssymb}
\usepackage{fancyhdr}
\usepackage{graphicx}
\newtheorem{thm}{Theorem}[section]
\newtheorem{cor}[thm]{Corollary}

\theoremstyle{definition}

\theoremstyle{remark}

\numberwithin{equation}{section}

\newcommand{\N}{\mathbb{N}}
\newcommand{\Z}{\mathbb{Z}}

\renewcommand{\phi}{\varphi}

\DeclareMathOperator{\rat}{Rat}

\newcommand{\pvali}{\vspace{12pt}}
\newcommand{\vali}{\vspace{6pt}}

\newcommand{\bin}{\{0,1\}}

\begin{document}
\title{Undecidability in Finite Transducers, Defense Systems and Finite
Substitutions}

\author{Vesa Halava\thanks{Supported by emmy.network foundation
under the aegis of the Fondation de Luxembourg.}
\\ Department of Mathematics and Statistics\\ University of Turku, Finland\\
Email: \texttt{vesa.halava@utu.fi}
}
\date{October 1997}

\maketitle

\begin{abstract}
In this manuscript we present a detailed proof for undecidability of the equivalence of finite substitutions on regular language $b\{0,1\}^*c$. The proof is based on the works of Leonid P. Lisovik. 
\end{abstract}

\section{Introduction and history}

This manuscript was written during the summer of 1997 while the author worked as a research assistant in Prof. Juhani 
Karhum\"aki's project. The task for the summer was to read and verify in details the proof of undecidability of the equivalence problem for finite substitutions on regular languages proved by Prof. Leonid P. Lisovik from Kiev, Ukraine. As a result the author wrote the present manuscript  based on articles~\cite{Lis:83, Lis:91,Lis:97}. In the original articles a lot of details were left to the reader. 

The main motivation for the manuscript was that Lisovik in~\cite{Lis:97} was able to prove that the equivalence problem problem for finite substitutions was undecidable already for a quite simple regular language $b\{01,1\}^* c$, see Section~\ref{sec:eq}. Lisovik's proof for this language was simplyfied by Halava and Harju \cite{HaHa:99} using the undecidability of the universe problem in integer weighted finite automata instead of the undecidability track of Lisovik's from the inclusion problem of finite transducers (Section~\ref{trans}) through undecidability in so called defence systems defined by Lisovik himself (Section~\ref{defe}). Note that the regular language with undecidable equivalence problem for finite substitutions was later improved by Karhum\"aki and Lisovik\footnote{It needs to be mentioned that Lisovik was a frequent visitor of Karhum\"aki's group in Turku around that time. Many stories of his peculiar but extremely friendly behaviour are still told in Turku.  The author remembers particularly well the party after the defence of his PhD thesis in April 2002 where Lisovik participated, not with any official role on the defence, but as a quest as he happened to visit Turku at that time: Lisovik gave altogether almost ten speeches during the dinner and the topics of these speeches varied somewher between math, life and basketball. For the sake of honesty it must be told that  after the first five speeches, Lisovik was encouraged by author's official supervisor Prof. Tero Harju to give more speeches. Naturally, the author is grateful for both, especially, because according to the official protocol of the party, the PhD candidate has to reply to all the speeches given with a new speech. } \cite{KaLi1} in 2002 (alternatively, see~\cite{KaLi2})  to the language $ab^*c$, and, further, by Kunc~\cite{Kunc} in 2007 to the language $a^*b$.   

As mentioned above, the root of undecidability in Lisovik's proof is the undecidability of the inclusion of two
rational relations (recognized by finite transducers), the result which was originally proved by Ibarra~\cite{Iba:78}. Lisovik gave a new proof for this result in 1983  (see~\cite{Lis:83}) with a clever reduction from the Post Correspondence Problem. Indeed, the main motivation for publishing this manuscript now 24 years later lays on this proof, as it has not been published in this form before. Recently, in \cite{handbook} Harju and Karhum\"aki presented a version of this proof with citation to this manuscript.

\section{Finite transducers}\label{trans}

Let $\Sigma$ be an alphabet and denote by $\epsilon$ the empty word.
The star operation on $\Sigma$, $\Sigma^*$, is as usual the set of
all word over $\Sigma$. Denote by $\Sigma^+=\Sigma^*\setminus \{\epsilon\}$.

We begin with a definition of {\it finite transducer}, FT for short,
which is a 6-tuple $(Q,\Sigma,\Delta,E,q_0,F)$, where

\begin{itemize}
\item $Q$ is a finite set of states,

\item $\Sigma$ and $\Delta$ are input and output alphabets,

\item $E\subseteq Q\times \Sigma^*\times \Delta^*\times Q$ is a finite
set of transitions,

\item $q_0\in Q$ is the initial state and $F\subseteq Q$ is the set of
final states.

\end{itemize}

FT is a finite automaton with output. If the underlying
automaton is nondeterministic, then FT is called {\it generalized
sequential machine}, GSM for short, or {\it sequential transducer}.

Let $T$ be a finite transducer. Define the set
\begin{align*}
O(T)=&\{(w,y)\mid w=a_0\dots a_n, \quad y=b_0\dots b_n, \quad n\in\N, \quad
a_i\in\Sigma^*,\\
&b_i\in \Delta^*,\quad 0\le i\le n,
\text{ and there exists states }q_i\in Q,\text{ such that }\\
&(q_i,a_i,b_i,q_{i+1})\in E \text{ and }q_{n+1}\in F\}
\end{align*}
If  $(w,y)\in O(T)$, then we say that $(w,y)\in \Sigma^*\times \Delta^*$ is
recognized by $T$.

Let
$$
L(T)=\{w\mid (w,y)\in O(T) \text{ for some }y \}
$$
be the language accepted by the finite transducer $T$.

A subset $O(T)$ of $\Sigma^*\times \Delta$, which is recognized by
a FT $T$ is called {\it rational relation}. We denote the family of rational
relations of $\Sigma^*\times \Delta^*$ by $\rat(\Sigma^*\times\Delta^*)$.

It is clear that if $A,B\in \rat (\Sigma^*\times \Delta^*)$, then
\begin{align*}
A&\cup B\quad\text{ and }\\ A\cdot B=AB&=\{(w_1w_2,y_1y_2)\mid
(w_1,y_1)\in O(A), (w_2,y_2)\in O(B)\}
\end{align*}
are in $\rat (\Sigma^*\times\Delta^*)$. The union is clear, since
we may connected the FT's that recognize $A$ and $B$, by merging
their initial states of FT's recognizing $A$ and $B$. The product
$AB$ is recognized, by an FT, where we define every final state of
FT recognizing $A$ to be a initial state of the FT recognizing $B$.

The star operation for subset $U$ of $\Sigma^*\times \Delta^*$ is
defined naturally by
$$
U^*=\bigcup_{i\ge 0} U^i,
$$
where $U^i$ is the $i$'th power of $U$ defined using the product by initial values $U^0=\{\epsilon\}\times \{\epsilon\}$,  
$U^1=U$, and $U^{i+1}=UU^i$ for all $i\ge 1$.

\vali
We shall next prove that the equivalence and inclusion of two
rational relations is an undecidable problem  in the case where
$\Delta$ is unary. This result has many proofs, for example cf.
\cite{Iba:78}, \cite{Lis:83}. We shall here present the construction
from \cite{Lis:83}.

Before the theorem, recall that the {\it Post Correspondence Problem}, PCP
for short, which asks for a given pairs of non-empty words over
alphabet $\Gamma$,
$(u_1,v_1),(u_2,v_2),\dots,(u_n,v_n)$, whether there exists a sequence
$$
1\le \alpha_1,\alpha_2,\dots, \alpha_s\le n
$$ such that
$$
u_{\alpha_1}u_{\alpha_2}\dots u_{\alpha_s}=v_{\alpha_1}v_{\alpha_2}\dots
v_{\alpha_2},
$$
is known to be an undecidable problem. For more details about the PCP, cf. \cite{Post:46}, \cite{morphisms}.

\begin{thm}\label{lause}
Let $A$ and $B$ be two rational relations from $\rat (\Sigma^*\times c^*)$.
Then it is undecidable, whether
\begin{align*}
1) \quad A&\subseteq B,
\\ 2) \quad A&=B.
\end{align*}
\end{thm}

\begin{proof}
Assume that $(u_1,v_1),\dots,(u_n,v_n)$ is a sequence of pairs of non-empty
words over
$\{a,b\}$. Define alphabet $\Sigma=\{a,b,i_1,\dots,i_n\}$, and
$k_\alpha=|u_\alpha|$ for all $\alpha=1,2,\dots,n$.

Next we define needed subsets of $\Sigma^+\times c^+$:
\begin{align*}
L_1&=\{(i_\alpha,c^{k_\alpha+1})\mid 1\le \alpha \le n\}^*,\\
L_2&=\bigcup_{\beta=1}^n\bigcup_{j=1}^{k_\beta} L_{\beta j}, \\
\intertext{where 
$L_{\beta j}=L_1\cdot (i_\beta,c^j)\{(i_\alpha,c)\mid 1\le \alpha\le n\}^*$,} 
L_3&=L_2\{(a,c),(b,c)\}^*,\\ 
L_4&=L_1\{(a,c),(b,c)\}^*
\{(a,c^2),(b,c^2)\}^+.
\end{align*}
Finally, for $\beta\in\{1,\dots,n\}$, let
\begin{align*}
S_\beta&=\{\mu\mid \mu\in \{a,b\}^*,|\mu|=|u_\beta|, \mu\ne u_\beta\},\\
\intertext{and set}
L_5&=\bigcup_{\beta=1}^n\bigcup_{\mu \in S_\beta} M_{\beta \mu},\\
\intertext{where }
M_{\beta \mu}&=L_1(i_\beta,c)\{(i_\alpha,c)\mid 1\le \alpha\le n\}^*\{(a,c),
(b,c)\}^* (\mu,c^{2k_\beta})\{(a,c^2),b,c^2)\}^*.\\
\end{align*}
Now we define
$$
L_u=L_3\cup L_4\cup L_5.
$$
Similarly, let $L_v$ be defined for the second components of the pairs $(u_\alpha, v_\alpha)$ in the sequence.  Note that $L_u$ and $L_v$ are in $\rat(\Sigma^*\times c^*)$, since we can define nondeterministic
FT's to recognize $L_1$, $L_{\beta j}$'s, $M_{\beta \mu}$'s and
therefore also $L_2$, $L_3$, $L_4$ and $L_5$ are rational
relations.

Next define
$L_0=\{(i_\alpha,c)\mid 1\le \alpha\le n\}^+\{(a,c^2),(b,c^2)\}^+$.
It is easy to construct a FT, that recognizes $L_0$.

\vali
{\it Claim.} $L_0\subseteq L_u\cup L_v$ if and only if
there does not exist sequence of $\alpha_i$'s, such that $1\le \alpha_1,\dots,\alpha_s\le n$
and
$u_{\alpha_1}\dots u_{\alpha_s}=v_{\alpha_1}\dots v_{\alpha_s}$.

\vali
{\it Proof of the Claim.} Assume that there exists such sequence $\alpha_1,
\dots,\alpha_s$, that PCP has solution and let
$$
w=(x,y)=(i_{\alpha_1}\dots i_{\alpha_s}u_{\alpha_1}\dots
u_{\alpha_s},c^{s+2(k_{\alpha_1}+
\dots+k_{\alpha_s})})\in L_0.
$$
(i) If $w\in L_3$, then for some
$
w_1=(i_{\alpha_1}\dots i_{alpha_s},c^m)\in L_2$,
$$
w=w_1(u_{\alpha_1}\dots u_{\alpha_s},c^{k_{\alpha_1}+
\dots+k_{\alpha_s}}).
$$
Therefore $w_1\in L_{\beta j}$, for some $\beta\in
\{\alpha_1,\dots,\alpha_s\}$ and $1\le j\le k_{\beta}$, and so in
path recognizing $w_1$, $i_{beta}$ has outputs $c^{j}$ and
$j<k_\beta+1$, so $m<k_{\alpha_1}+\dots+k_{\alpha_s}+s$. Therefore
$w\notin L_3$.

(ii) Let $\beta_i$'s, $i\in \{1,r\}$ be a sequence such that
$1\le \beta_1,\dots ,\beta_r\le n$, and let
$$
w_1=(i_{\alpha_1}\dots i_{\alpha_s}u_{\beta_1}\dots
u_{\beta_r},c^m)\in L_4.
$$
In the recognizing paths of $w_1$, for each $i_{\alpha_j}$ the
output is $k_{\alpha_j}+1$ and for $u_j$,
$j\in\{\beta_1,\dots,\beta_r\}$, the output is $c^{\ell_j}$, where
$\ell_j\ge k_j$ and at least for one $j$ $\ell_j >k_j$, because of
$\{(a,c^2),(b,c^2)\}^+$. So we have that
$m>k_{\alpha_1}+\dots+k_{\alpha_s}+s+k_{\beta_1}+\dots+k_{\beta_r}$,
and therefore $w\notin L_4$.

(iii) Assume that $w\in L_5$. Then there exists integer $\beta$,
$1\le \beta \le s$, and $\gamma_1,\mu,\gamma_2\in \{a,b\}^*$ such that
$w\in M_{\alpha_\beta \mu}$, $u_{\alpha_1}\dots u_{\alpha_s}=\gamma_1\mu\gamma_2$,
$|\mu|=u_{\alpha_\beta}$ and $\mu\ne u_{\alpha_\beta}$. If
$$
(i_{\alpha_1}\dots i_{\alpha_\beta}\dots
i_{\alpha_s}\gamma_1\mu\gamma_2,c^m)\in M_{\alpha_\beta \mu},
$$
then
\begin{align*}
m&=(k_{\alpha_1}+1)+\dots
+(k_{\alpha_{\beta-1}}+1)+(s-\beta+1)+|\gamma_1|+2|\mu|+
2|\gamma_2|\\ &=k_{\alpha_1}+\dots+k_{\alpha_{\beta
-1}}+s+|\gamma_1|+2|\mu|+2|\gamma_2|.
\end{align*}
Now since $w\in L_5$ and $|\gamma_1\mu\gamma_2|=k_{\alpha_1}+\dots+
k_{\alpha_s}$, we get that
$$
|\mu|+|\gamma_2|=k_{\alpha_\beta}+\dots+k_{\alpha_s},
$$
and since $|\mu|=k_{\alpha_\beta}$, finally
$$
|\gamma_2|=k_{\alpha_{\beta+1}}+\dots+k_{\alpha_s}\text{ and }
|\gamma_1|=k_{\alpha_1}+\dots+k_{\alpha_{\beta-1}}.
$$
It follows that $\mu=u_{\alpha_\beta}$ and we have a contradiction.
Therefore $w\notin L_5$.

So $w\notin L_u$ and by similarly it can shown that $w\notin L_v$,
and we have proved one direction of the claim.

\vali
Assume now that there is no sequence $1\le
\alpha_1,\dots,\alpha_s\le n$ such that the instance of PCP has
solution. Let $w_1\in \{a,b\}^+$ and $w=(i_{\alpha_1}\dots
i_{\alpha_s}w_1,c^{s+2|w_1|})\in L_0$.

By assumption, $w_1\ne u_{\alpha_1}\dots u_{\alpha_s}$ or
$w_1\ne v_{\alpha_1}\dots v_{\alpha_s}$.
We shall
show that if $w_1\ne u_{\alpha_1}\dots u_{\alpha_s}$, then
$w\in L_u$. Of course then similarly, if
$w_1\ne v_{\alpha_1}\dots v_{\alpha_s}$, then $w\in L_v$.

(i) If $|w_1|>|u_{\alpha_1}\dots u_{\alpha_s}|$, i.e.
$|w_1|>k_{\alpha_1}+\dots+k_{\alpha_s}$, then
for some
$x,y\in \{a,b\}^+$, $|x|=k_{\alpha_1}+\dots+k_{\alpha_s}$, $w_1=xy$ and
$$
w=(i_{\alpha_1}\dots
i_{\alpha_s},c^{k_{\alpha_1}+\dots+k_{\alpha_s}+s})
(x,c^{k_{\alpha_1}+\dots+k_{\alpha_s}})(y,c^{2|y|})\in L_4.
$$

(ii) If $|w_1|<|u_{\alpha_1}\dots u_{\alpha_s}|$, i.e.
$|w_1|<k_{\alpha_1}+\dots+k_{\alpha_s}$, then there exists
$\beta\in\{1,\dots,s\}$ and $ j\in \{1,\dots,k_{\alpha_\beta}\}$
such that
$$
|w_1|=k_{\alpha_1}+\dots+k_{\alpha_{\beta-1}}+j-1,
$$
and so
\begin{align*}
w=&(i_{\alpha_1}\dots
i_{\alpha_{\beta-1}},c^{(k_{\alpha_1}+1)+\dots+(k_{\alpha_{\beta-1}}+1)})
(i_{\alpha_\beta},c^j)\\ &\cdot(i_{\alpha_{\beta+1}}\dots
i_{\alpha_s},c^{s-\beta})(w_1,c^{k_{\alpha_1}+\dots+k_{\alpha_{\beta-1}}+j-1})\in
L_3.
\end{align*}

(iii) If $|w_1|= k_{\alpha_1}+\dots+k_{\alpha_s}$, then since
$w_1\ne u_{\alpha_1}\dots u_{\alpha_s}$, there exists
$\beta\in\{1,\dots,s\}$ and $\mu,\gamma\in \{a,b\}^*$ such that
$$
w_1=u_{\alpha_1}\dots u_{\alpha_{\beta-1}}\mu\gamma \text{ and }
|\mu|=|u_{\alpha_\beta}| \text{ but } \mu \ne u_{\alpha_\beta},
$$
and so
\begin{align*}
w=&(i_{\alpha_1}\dots
i_{\alpha_{\beta-1}},c^{(k_{\alpha_1}+1)+\dots+(k_{\alpha_{\beta-1}}+1)})
(i_{\alpha_\beta},c)\\ &\cdot(i_{\alpha_{\beta+1}}\dots
i_{\alpha_s},c^{s-\beta}) (u_{\alpha_1}\dots
u_{\alpha_{\beta-1}},c^{k_{\alpha_1}+\dots+k_{\alpha_{\beta-1}}})
(\mu,c^{2|u_{\alpha_\beta}|})(\gamma,c^{2|\gamma|})\in L_5.
\end{align*}

So $w\in L_u$, if $w_1\ne u_{\alpha_1}\dots u_{\alpha_s}$ and so
the claim is proved.

\vali
Now by the undecidability of PCP, it is undecidable whether
$L_0\subseteq L_u\cup L_v$ and whether $L_0\cup L_u\cup L_v =
L_u\cup L_v$. This proves the theorem.
\end{proof}

\begin{cor}\label{cor:lis}
It is undecidable for two rational relations $A$ and $B$ from $\rat
(\bin^*\times c^*)$, whether
\begin{align*}
1)\quad A&\subseteq B,\\ 2) \quad A&=B.
\end{align*}
\end{cor}

\begin{proof}
Claim follows straight forwardly from Theorem \ref{lause}, since we
can encode the alphabet $\Sigma$ into $\bin^*$ and the result
remains.
\end{proof}

We shall next define a special type of finite transducer, so called
{\it $Z$-transducer}. FT $T$ is called $Z$-transducer, if it is of the
form
$$
(Q,\{0,1\},\{c, cc\},E,q_0,g_f),
$$
i.e. it has input alphabet $\{0,1\}$,
output alphabet $\{c\}$, only one final state $q_f$ and the set of
transitions $E\subseteq Q\backslash \{q_f\}\times \{0,1\}\times\{c,cc\}\times
Q$. We shall define $Z$-transducer as quadruple $(Q,E,q_0,g_f)$ from now on,
since input and output alphabets are fixed.
Notice that $Z$-transducer reads one symbol at a time and always outputs
one or two $c$'s. Notice also that there is no transitions from the final
state $q_f$ in $Z$-transducer.

A $Z$-transducer is called {\it deterministic} if the underlying
automaton is deterministic, i.e. if for any $a\in \{0,1\}$, $q\in
Q\backslash \{q_f\}$ there exists a unique transition $(q,a,b,p)$,
where $b\in \{c,cc\}$ and $p\in Q$. A $Z$-transducer is called {\it
complete}, if for any $a\in\{0,1\}$, $q\in Q\backslash\{q_f\}$
there exists at least one transition of the form $(q,a,b,p)$. Note
that here determinism preserves completeness. Note also that every
$Z$-transducer can be maid complete by adding a {\it garbage state}
$f$ into $Q$ such that if there does not exists any transition
$(q,a,b,p)$ for some $q$ and $a$, then we add transition
$(q,a,c,f)$ to $E$ and further we add transition $(f,a,c,f)$ to $E$
for $a\in\bin$.

Let $T$ be a $Z$-transducer. As for FT's, we define the set
\begin{align*}
O(T)=&\{(w,y)\mid w=a_0\dots a_n, \quad y=b_0\dots b_n, \quad n\in\N, \quad
a_i\in\{0,1\},\\
&b_i\in \{c,cc\},\quad 0\le i\le n,
\text{ and there exists states }q_i\in Q,\text{ such that }\\
&(q_i,a_i,b_i,q_{i+1})\in E \text{ and }q_{n+1}=q_f\}
\end{align*}

Note that
in deterministic $Z$-transducer $T$, for all $w\in\bin^*$,
there exists either a unique path when reading word $w$ or a prefix
$u$ of $w$ such that $u\in L(T)$. Since there is no transitions from
final state, we see that if $w=uv$, $v$ is a nonempty word, then
$u\in L(T)$ implies $w\notin L(T)$.

\begin{cor}\label{lemma:1}
Let $C$ and $D$ be two $Z$-transducers, $C$ is deterministic and
$D$ nondeterministic and complete. It is undecidable, whether
$O(C)\subseteq O(D)$.
\end{cor}

\begin{proof}
In the proof of Corollary \ref{cor:lis} we mentioned the coding of
the alphabet $\Sigma$ in Theorem \ref{lause} to binary alphabet.
Let $(u_1,v_1),\dots,(u_n,v_n)$ be the instance of PCP used in the
proof of Theorem \ref{lause}. We can for example use coding
$\delta$, where $k=1+\max_{1\le i\le n}\{|u_i|,|v_i|\}$ and
alphabet $\Sigma=\{a,b,i_1,\dots,i_n\}$ is encoded to set
$\{10^{i}1 \mid k\le i \le k+n+1\}$.

If we now code with $\chi$ each element $w=(v,c^m)\in\Sigma^+
\times c^+$ used in the proof of Theorem \ref{lause} in such a way
that $\chi(w)=(\delta(v)0,c^{m+|\delta(v)0|})$. Denote by
$\chi(L_i)$ the coded set $L_i$, $i=1,2,3,4,5,u,v,0$.

Clearly $\chi(L_1)$ can be reorganized by a non-deterministic
$Z$-transducer, when reading $\delta(i_\alpha)$ the transducer
outputs $cc$ for $k_\alpha+1$ first input symbols and $c$ for the
others. When reading the last 0 in the input, $Z$-transducer
outputs one $c$ and moves to final state.

Using the same idea, also other $\chi(L_i)$'s can be recognized by
a non-deterministic $Z$-transducer. Actually
$$
\chi(L_0)=\{(\delta(i_\alpha),c^{|\delta(i_\alpha)|+1})\mid 1\le \alpha\le n\}^+
\{(\delta(a),c^{|\delta(a)|+2}),(\delta(b),c^{|\delta(b)|+2})\}^+(0,c)
$$
can be reorganized by a deterministic $Z$-transducer.

Now since $\chi(L_0)\subseteq \chi(L_u)\cup \chi(L_v)$ if and only
if $L_0\subseteq L_u\cup L_v$, the claim follows by the proof of
Theorem \ref{lause}.
\end{proof}

We shall use result in above corollary in the next section.

\section{Defense Systems}\label{defe}

In this section we shall consider so called {\it defense systems}, DS for short.
Result in this section is from \cite{Lis:91}.
A DS system is intended to defense some elements of the set integers $\Z$. The
elements of $\Z$ are also called defense {\it nodes}.
Any DS is a triple $V=(K,H,\Gamma)$, where $K$ is set of {\it lines},
$$
K=\{i\mid 1\le i \le s,\quad i,s\in \Z\},
$$
$H$ is the set of instructions and $\Gamma$ is the set of attacking symbols.

Each node can be defended by lines from $K$. In other words, each node
can be defended by $s$ different lines. The initial situation in our case
is that
only node 0 is defended by line 1,
and the other nodes don't have defence at all.

The attacking system is supposed to `send' symbols from the set $\Gamma$ to
the defending system. This means that attacks can be thought as a words from
$\Gamma^*$.

Each rule of the set $H$ is of the form $(k,a,j,z,p)$, where $1\le k,j\le s$,
$a\in \Gamma$, $z\in \{-1,0,1\}$ and $p$ is the real number $0\le p\le 1$.
Each rule means that when attacking symbol $a$ is send, defense of node $i$
by line $k$ is transferred with
probability $p$ to defense of node $i+z$ by line $j$. We shall denote the
probability above also $p_{a,k,j}^z$. Naturally for all $a\in \Gamma$
$$
\sum_{j=1}^s\sum_{z=-1}^1 p_{a,k,j}^z=1,
$$
i.e. on each attacking symbol something necessarily happens.
Note that the underlying system in defense systems is nondeterministic and
therefore the model of defense systems we defined is sometimes called
{\it nondeterministic DS}, NDS for short.

\begin{figure}[ht]
\begin{center}
\unitlength1mm
\begin{picture}(100,50)
\multiput(30,40)(0,-5){3}{\line(1,0){60}}
\put(30,15){\line(1,0){60}}
\multiput(26,40)(0,-5){3}{\line(1,0){2}}
\multiput(23,40)(0,-5){3}{\line(1,0){1}}
\multiput(21,40)(0,-5){3}{\line(1,0){.5}}
\multiput(92,40)(0,-5){3}{\line(1,0){2}}
\multiput(96,40)(0,-5){3}{\line(1,0){1}}
\multiput(98.5,40)(0,-5){3}{\line(1,0){.5}}

\put(26,15){\line(1,0){2}}
\put(23,15){\line(1,0){1}}
\put(21,15){\line(1,0){.5}}
\put(92,15){\line(1,0){2}}
\put(96,15){\line(1,0){1}}
\put(98.5,15){\line(1,0){.5}}

\multiput(40,40)(10,0){5}{\makebox(0,0)[c]{$\bullet$}}
\multiput(40,35)(10,0){5}{\makebox(0,0)[c]{$\bullet$}}
\multiput(40,30)(10,0){5}{\makebox(0,0)[c]{$\bullet$}}
\multiput(40,15)(10,0){5}{\makebox(0,0)[c]{$\bullet$}}
\put(18,40){\makebox(0,0)[c]{1}}
\put(18,35){\makebox(0,0)[c]{2}}
\put(18,30){\makebox(0,0)[c]{3}}
\put(18,15){\makebox(0,0)[c]{$s$}}
\put(40,45){\makebox(0,0)[c]{$-2$}}
\put(50,45){\makebox(0,0)[c]{$-1$}}
\put(60,45){\makebox(0,0)[c]{$0$}}
\put(70,45){\makebox(0,0)[c]{$1$}}
\put(80,45){\makebox(0,0)[c]{$2$}}
\put(60,40){\circle{3}}
\put(60,26){\makebox(0,0)[c]{$\cdot$}}
\put(60,23){\makebox(0,0)[c]{$\cdot$}}
\put(60,20){\makebox(0,0)[c]{$\cdot$}}

\end{picture}
\caption{\label{inids} A picture illustrating a defense system in the initial configuration
defending the node 0 by the line 1.}
\end{center}
\end{figure}
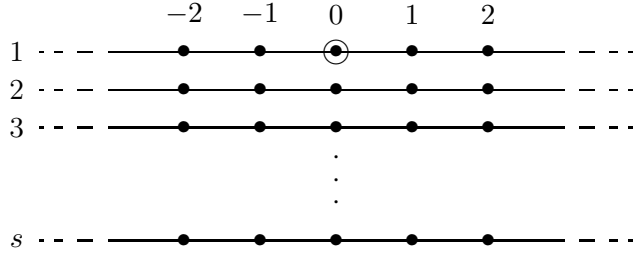

We fix the attacking symbol set $\Gamma=\bin$ in this paper.

\vali
A NDS can also be viewed as a countable Markov system. To simplify notations
we denote each configuration of a NDS by an integer. If node $i$ is defended
by line $j$, we denote this configuration by integers $i\cdot s+(j-1)$.
Recall that the initial configuration is that node 0 is defended by line
1, which is represented as an integer 0.

Let $w\in\bin^*$. We shall denote the probability that the NDS is in the
configuration $k\in \Z$ in response to a finite sequence of attacking signals
$w$ by $p_w(k)$.

\vali
Let $D=(K,H,\Gamma)$ be a defense system. $D$ is called {\it unreliable} if,
for some $w\in\Gamma^*$,
after attacking sequence $w$ the probability that node 0 is defended
by some line is 0, i.e. $p_w(j)=0$ for all $0\le j \le s-1$. The word $w$ here
is called {\it critical}. If there is no
critical words $w\in \Gamma^* $ that $D$,
then $D$ is called {\it reliable}.

\begin{thm}\cite{Lis:91}\label{thm:1}
The unreliability of NDS is undecidable, i.e. it is undecidable for a given
NDS $B=(K,H,\bin)$, to determine whether there exist $w\in\bin^*$ such that
$p_w(j)=0$ for all $0\le j\le s-1$.
\end{thm}

\begin{proof}
In this proof we shall use the undecidability result of Corollary
\ref{lemma:1}.

Let $C$ be a deterministic $Z$-transducer and $D$ be a nondeterministic and
complete $Z$-transducer, $C=(K_1,H_1,q_0,q_f)$ and $D=(K_2,H_2,g_0,g_f)$.
Define a nondeterministic complete $Z$-transducer $D'=(K_3,H_3,g_0,g_f)$,
where
\begin{align*}
K_3&=K_1\cup K_2,\\
H_3&=H_1\cup H_2\cup \{(g_0,a,b,q)\mid (q_0,a,b,q)\in H_1)\}.
\end{align*}
$Z$-transducer $D'$ satisfies $O(D')=O(D)$, but the transducer also has
paths of $C$ in it, although they are not accepting paths.

Let $s$ be the number of elements of the set
$$
K=K_1\times K_3=\{(q,g)_j\mid 1\le j\le s\} \text{ and }(q,g)_1=(q_0,g_0).
$$
Let
$$
H\subseteq K\times\bin\times \{-1,0,1\}\times K
$$
so that $((q_k,g_\ell)_i,a,z,(q_r,g_t)_j)\in H$, if
$$
(q_k,a,b_1,q_r)\in H_1 \text{ and } (g_\ell,a,b_2,g_t)\in H_3,
$$
and $z$ follows by the rules $(b_1,b_2\in \{c,cc\})$
\begin{equation}\label{eq:z}
z=\begin{cases} -1 &\text{if }b_1=b_2c,\\
                0 &\text{if }b_1=b_2,\\
                1 &\text{if }b_2=b_1c.
\end{cases}
\end{equation}
Moreover $H$ contains elements
\begin{equation}\label{eq:1}
((q_f,g_f),a,0,(q_f,g_f)),
\end{equation}
\begin{equation}\label{eq:2}
((q,f),a,1,(q_f,q_f)), \text{ where }\{q,f\}\cap \{q_f,g_f\}\ne \emptyset,\quad
a=0,1.
\end{equation}.
We shall refer the elements of $H$  as rules.

We shall now associate a NDS B to construction above. Let
$$
M_{a,k}^z=\{j\mid ((q,g)_k,a,z,(q,g)_j)\in H\}
$$
and let $m(a,k,z)=\vert M_{a,k}^z\vert$ and
$$
m(a,k)=\sum_{z=-1}^1 m(a,k,z).
$$
Let $B=(K',H',\bin)$ be a defense system, such that $K'=\{1,\dots,s\}$,
if $((q,g)_k,a,z,(q,g)_j)\in H$, then $(k,a,z,j,p_{a,k,j}^z)\in H'$,
$p_{a,k,j}^z=1/m(a,k)$. This probability is obvious by the construction.

\pvali
{\it Claim.} The existence of finite sequence $w\in\bin^*$ such
that the NDS $B$ has $p_w(j)=0$ for all $0\le j\le s-1$ is equivalent to the
fact that $O(C)\not\subseteq O(D)$.

\vali
Before the proof, we note few facts about the construction.
Our defense system $B$ simulates the calculations of $Z$-transducers
$C$ and $D'$ at a same time in its lines, which can be thought as an elements
of $K=K_1\times K_3$.

By \eqref{eq:z}, $z$ gives the difference of lengths
of outputs in $C$ and $D'$. It follows that if the defended node is 0,
the outputs of $C$ and $D'$ are equal. If the node is negative, the length of
the output of $C$ is larger than the length of the output of $D'$
by the absolute value of the node. If it is positive,
then vice versa.

Now we are ready to proof the equivalence mentioned above.

\pvali
{\it Proof of the Claim.}
Assume first that $O(C)\not\subseteq O(D)$. This means that there
exists a word $w\in \bin^*$ such that for unique $y\in c^*$,
$(w,y)\in O(C)$, but $(w,y)\notin O(D)$. We have two cases:

\vali
i) If $w\in L(D)$, then for all $(w,y')\in O(D)$, $y'\ne y$. There exists
four kind of paths in our NDS $B$, that have positive probability on attacking
sequence $w$, we separate them in terms of calculations of $C$ and $D'$:

\vali
1) If the simulation of $D'$ is similar to simulation of $C$. Then
we are all the time defending the node 0 and end up in state $(q_f,q_f)$.
Now for a word $wa$, $a\in\bin$, we use rule \eqref{eq:2} and the defense
shifts to node 1, since $z=1$. Note that we can add several symbols to
$w$, and defense of node moves to one larger by every symbol. The simulation
of $C$ does not change from beginning, since no
subword of accepted word can be accepted in deterministic $Z$-transducer.

2) If the simulation of $D'$ reaches the final state $g_f$ before than the
simulation of $C$. After that the rule used is \eqref{eq:2}.
Every step of this rule moves the defense of the node to the node
one larger. After that we may add a symbols from $\bin$ to the end of
$w$ to get the defense to a positive node.

3) If simulation of $D'$ is not in the final state when the simulation
$C$ ends. Again after that we may add symbols of $\bin$ to the end
of the word $w$ to get the defense to a positive node.

4) If the simulations of $D'$ and $C$ reach the final state at the same time, i.e.
in the end of $w$. Of course the node defended at that time can't be
0, since then the outputs would be equal in $C$ and $D$, and that is
impossible, by the fact that $(w,y)\notin O(D)$.
We can again add symbols to end of $w$, and the rule used is \eqref{eq:1}
and that does not change the defense anywhere.

\vali
By cases 1-4, we see that, there exists a word $wv$, $v\in\bin^*$ such that
$p_{wv}=0$ for all $0\le j\le s-1$. This follows, since there is a limit for
symbols, that has to be added to get all these possible paths of defense
to positive nodes.

\vali
ii) If $w\notin O(D)$, then the 1-3 above cases are possible, and again
there exists $wv$, $v\in\bin^*$, such that NDS $B$ is unreliable.

So we have proved that if $O(C)\not\subseteq O(D)$, then NDS
$B$ is unreliable.

\pvali
Assume next that NDS $B$ is unreliable. It means that
there exists sequence $w\in\bin^*$ such that
$p_w(j)=0$ for all $0\le j\le s-1$. By the fact that $C$ is deterministic
and therefore complete, it means that some subword of $w$ must be in
$L(C)$, since otherwise there is a path in $C$ for a input word $w$
and therefore in $B$ node 0 has positive defense probability for
some line, which is related to element $(q,q)\in K$, $q\in K_1$.

Now assume that $v$ is the subword of $w$ such that $v\in L(C)$ and
let $y$ be the unique element of $\bin^*$, such that $(v,y)\in O(C)$.
Now $(v,y)\notin O(D)$, since otherwise there would be a possible defense
in node 0 after attacking sequence $v$ and after $v$ the instruction used
would be the corresponded to rule \eqref{eq:1} which does not move the
defense anywhere. Therefore for the attacking sequence $w$ there would be
a defense in the node 0 with positive probability,
which is not possible by the assumption.

Now we have finally proved the Claim.

\pvali
By Corollary \ref{lemma:1} it is undecidable whether
$O(C)\not\subseteq O(D)$ and therefore the unreliability of NDS is
also undecidable.
\end{proof}

Note that since unreliability is a complement of reliability, this also
means that reliability is undecidable.

\section{Finite substitutions}\label{sec:eq}

Let $\Sigma$ and $\Delta$ two alphabets. For a set $S$ denote by $2^{S}$ the power set
of $S$, i.e. the collection of all subsets of $S$.

A mapping $\phi:\Sigma^* \to 2^{\Delta^*}$ is called {\it substitution}, if

\vali
1) $\phi(\epsilon)=\{\epsilon\}$ and

2) $\phi(xy)=\phi(x)\phi(y)$.

\vali
Because of condition 2, a substitution is usually defined by giving
the images of all letters in $\Sigma$.

Let $\phi$ be as above and $L$ be a language over $\Sigma^*$, i.e. $L\subseteq \Sigma^*$. We denote
$$
\phi(L)=\bigcup_{w\in L} \phi(w).
$$
Two substitutions $\phi,\xi:\Sigma^*\to 2^{\Delta^*}$ are equivalent on language
$L$ if
$$
\phi(L)=\xi(L).
$$

A substitution $\phi$ is called $\epsilon${\it -free}, if $\epsilon\notin\phi(a)$
for all $a\in \Sigma$.
And it is called a {\it finite substitution} if, for all $a\in\Sigma$,
the set $\phi(a)$ is finite.

\pvali
A language $L$ is called {\it regular}, if it is accepted by a
finite automaton. It is known that regular languages are closed under finite
substitutions, which means that if $L$ is regular , so is $\phi(L)$ for finite
substitution $\phi$.

Next theorem states an undecidability result concerning finite substitutions and
regular languages. It is from \cite{Lis:97}

\begin{thm}
The equivalence problem for $\epsilon$-free finite substitutions on regular
language $b\{0,1\}^*c$ is undecidable.
\end{thm}

\begin{proof}
We shall use Theorem \ref{thm:1}. Let $V=(K,H,\bin)$ be a NDS defined
in previous section, $K=\{1,..,s\}$, $H$ is the set of instructions and
attacking symbol set is $\bin$.
We shall define two finite substitutions
$\phi,\xi:\{b,0,1,c\}^*\to \bin^*$ such that $\phi$ and $\xi$ are equivalent
on language $b\bin^*c$ if and only if NDS $V$ is reliable.

First we define following sets and words:
\begin{align*}
D_a&=\{(k,z,j)\mid (k,a,j, z,p)\in H \text{ for some }p>0\}, \quad a\in\bin,\\
D&=D_0\cup D_1,\\
w&=010010001\dots10^{s+1}1,\quad w^0=\epsilon, \quad w^1=w,\quad w^2=ww,\\
\alpha_k&=01001\dots10^k1,\quad \beta_k=0^{k+1}1\dots10^{s+1}1, \text{ for }  1\le k\le s,\\
w&=\alpha_k\beta_k, \quad F(k,z,j)=\beta_k w^{z+1}\alpha_j ,\quad
F(k,z)=F(k,z,j)\beta_j=\beta_kw^{z+2},\\ T_a&=\bigcup_{(k,z,j)\in
D_a} \{F(k,z,j)\},\quad C_a=\bigcup_{(k,z,j)\in D_a} \{F(k,z)\},
\quad a\in\bin \\
C&=C_0\cup C_1,\:  M=\{w\},\:  B=\{ww\},\: N=\{\beta_k\mid 1\le
k\le s\},\text{ and } S=\{\alpha_1\}.
\end{align*}

Now we can define finite substitutions $\phi,\xi:\{b,0,1,c\}\to \bin^*$:
\begin{align*}
\xi(b)&=S\cup MN=\{\alpha_1\}\cup \{w\beta_k \mid 1\le k\le s\},\\
\phi(b)&=\xi(b)\cup M=\{\alpha_1\}\cup \{w\beta_k \mid 1\le k\le s\}\cup \{w\},\\
\xi(c)&=\phi(c)=M\cup NM=\{w\}\cup \{\beta_kw \mid 1\le k\le s\},\\
\xi(a)&=\phi(a)=B\cup T_a\cup NT_a\cup C_aN\cup NC_aN\\
&=\{ww\}\cup \{\beta_kw^{z+1}\alpha_j\mid (k,z,j)\in D_a\}\\
&\cup \{\beta_\ell\beta_k w^{z+1}\alpha_j\mid 1\le \ell\le s,(k,z,j)\in D_a\}\\
&\cup \{\beta_kw^{z+2}\beta_\ell\mid 1\le \ell\le s,(k,z,j)\in D_a\}\\
&\cup \{\beta_{\ell_1}\beta_kw^{z+2}\beta_{\ell_2}\mid 1\le \ell_1,\ell_2\le s,
(k,z,j)\in D_a\},
\end{align*}
for $a=0,1$.
Let $L$ be the language $b\bin^*c$. Now clearly $\xi(x)\subseteq \phi(x)$ for
all $x\in L$, since $\xi(a)\subseteq \phi(a)$ for all letters
$a\in \{b,0,1,c\}$. Therefore to prove that $\xi(L)=\phi(L)$ iff and only
iff $V$ is reliable, we have show
that $\phi(L)\in\xi(L)$ iff and only iff $V$ is reliable.

\vali
Suppose first that $V$ is reliable. Let $x=x_0\dots x_{n+1}\in L$,
$u=u_0\dots u_{n+1}$, where $x_i\in \{b,0,1,c\}$ and $u_i\in \phi(x_i)$ for
all integers $0\le i\le n+1$. Note that $x_0=b$ and $x_{n+1}=c$. We have to
show that there exists $v_i\in \xi(x_i)$ for all $0\le i\le n+1$ such that
$v=v_0\dots v_{n+1}=u$.

First we note that the only difference in images by $\xi$ and $\phi$ is
in images of $b$, and $\xi(b)\setminus \phi(b)=M$. Therefore,if $u_0\ne w$, we have trivial solution $u_i=v_i$ for all $0\le i\le n+1$.
So we assume that $u_0=w$.

We shall use parenthesis to illustrate
factorizations by $\phi$ and $\xi$ to $u_i$'s and $v_i$'s. Now we divide
into three cases:

\vali
(i) If $n=0$, then $x=bc$ and we have two cases:

1) If $u_1=w\in \phi(c)$, then $u_0u_1=(w)(w)=(\alpha_1)(\beta_1w)\in\xi(x)$.

2) If $u_1\in NM\subseteq\phi(c)$, i.e. for some $1\le k\le s$,
$u_0u_1=(w)(\beta_kw)=(w\beta_k)(w)\in\xi(x)$.

\vali
(ii) If $n\ge 1$ and $u_1\notin B$. We shall show that there is a
factorization such that $u_i=v_i$ for $2\le i\le n+1$ and $u_0u_1=v_0v_1$.
Here we have four cases:

1) If $u_1\in T_{x_1}$, then, for $(k,z,j)\in D_{x_1}$,
$$
u_0u_1=(w)(\beta_kw^{z+1}\alpha_j)=(\alpha_1)(\beta_1\beta_kw^{z+1}\alpha_j)
=v_0v_1, \quad v_0\in S,v_1\in NT_{x_1}.
$$

2) If $u_1\in NT_{x_1}$, then, for $(k,z,j)\in D_{x_1}$ and $1\le \ell\le s$,
$$
u_0u_1=(w)(\beta_\ell\beta_kw^{z+1}\alpha_j)=(w\beta_\ell)(\beta_kw^{z+1}\alpha_j)
=v_0v_1, \quad v_0\in MN,v_1\in T_{x_1}.
$$

3) If $u_1\in C_{x_1}N$, then, for $(k,z,j)\in D_{x_1}$ and $1\le \ell \le s$,
$$
u_0u_1=(w)(\beta_kw^{z+2}\beta_\ell)=(\alpha_1)(\beta_1\beta_kw^{z+2}\beta_\ell)
=v_0v_1, \quad v_0\in S,v_1\in NC_{x_1}N.
$$

4) If $u_1\in NC_{x_1}N$, then, for $(k,z,j)\in D_{x_1}$ and $1\le \ell,t \le s$,
$$
u_0u_1=(w)(\beta_\ell\beta_kw^{z+2}\beta_t)=(w\beta_\ell)(\beta_kw^{z+2}\beta_t)
=v_0v_1, \quad v_0\in MN,v_1\in C_{x_1}N.
$$

\vali
(iii) If $n\ge 1$ and $u_1\in B$, then we need the reliability of $V$. Let
$t=\min\{i\mid i\ge1, u_i\notin B\}$. So the word
$u_0u_1\dots u_{t-1}=w(ww)\dots (ww) = w^{2t-1}$.

Since $V$ is reliable, there exists for attacking sequence
$x'=x_1\dots x_{t-1}\in \bin^*$ a sequence
$$
(j_0=1,x_1,j_1,z_1,p_1)(j_1,x_2,j_2,z_2,p_2)\dots
(j_{t-2},x_{t-1},j_{t-1},z_{t-1},p_{t-1})
$$
of elements of $H$ such that $p_i>0$ for all $1\le i\le t-1$ and
\begin{equation}\label{eq:3}
\sum_{i=1}^{t-1} z_i=0.
\end{equation}
Therefore there exists a sequence
$$
(j_0=1,z_1,j_1)(j_1,z_2,j_2)\dots
(j_{t-2},z_{t-1},j_{t-1}),
$$
where $(j_{i-1},z_i,j_i)\in D_{x_1}$. Now define $v_0=\alpha_1$, and
for $1\le i\le t-1$,
$$
v'_i=\beta_{j_{i-1}}w^{z_i+1}\alpha_{j_i}\in T_{x_i},
$$
we get that
\begin{align*}
v_0v'_1\dots v'_{t-1}&=\alpha_1\beta_1w^{z_1+1}\alpha_{j_1}
\beta_{j_1}w^{z_2+1}\alpha_{j_2}\dots \beta_{j_{t-2}}w^{z_{t-1}+1}\alpha_{j_t-1}\\
&=ww^{z_1+1}ww^{z_2+1}w\dots ww^{z_{t-1}+1}\alpha_{j_{t-1}}.
\end{align*}
Now by \eqref{eq:3} we get that
$$
v_0v'_1\dots v'_{t-1}=w^{2t-2}\alpha_{j_{t-1}}.
$$
So we have that $u_0u_1\dots u_{t-1}=v'_0v'_1\dots v'_{t-1}\beta_{j_{t-1}}$.
We may already set $v_i=v'_i$ for $1\le i\le t-2$.

Now we have two cases depending on $t$.
First if $t=n+1$, then
we have two cases:

1) If $u_{n+1}\in M$, then $v_{t-1}=v'_{t-1}$
and $v_{n+1}=\beta{j_{t-1}}w\in NM$ and
so $u=v$.

2) If $u_{n+1}\in NM$, $u_{n+1}=\beta_kw$, then we set
$$
v_{t-1}=v_n=\beta_{j_{t-2}}w^{z_{t-1}+2}\beta_k\in C_{x_n}N
\quad \text{ and } v_{n+1}=w\in M.
$$
Again $u=v$.

Second case is that $t\le n$. Then we set $v_i=u_i$ for $t+1\le i\le n+1$
and so we have four cases for $v_t$ and $v_{t-1}$:

1) If $u_{t}\in T_{x_t}$, for some $(k,z,j)\in D_{x_t}$
$u_t=\beta_kw^{z+1}\alpha_j$, then we set
$$
v_{t-1}=v'_{t-1} \text{ and }
v_t=\beta_{j_{t-1}}\beta_kw^{z+1}\alpha_j\in NT_{x_t},
$$
to get $u=v$.

2) If $u_{t}\in NT_{x_t}$, for some $(k,z,j)\in D_{x_t}$, $1\le \ell\le s$,
$u_t=\beta_\ell\beta_kw^{z+1}\alpha_j$, then we set
\begin{align*}
v_{t-1}&=\beta_{j_{t-2}}w^{z_{t-1}+1}\alpha_{j_{t-1}}\beta_{j_{t-1}}\beta_\ell
=\beta_{j_{t-2}}w^{z_{t-1}+2}\beta_\ell\in C_{x_t}N\\
\intertext{and}
v_t&=\beta_kw^{z+1}\alpha_j\in T_{x_t},
\end{align*}
to get $u=v$.

3) If $u_{t}\in C_{x_t}N$, for some $(k,z,j)\in D_{x_t}$, $1\le \ell\le s$,
$u_t=\beta_kw^{z+2}\beta_\ell$, then we set
$$
v_{t-1}=v'_{t-1} \text{ and }
v_t=\beta_{j_{t-1}}\beta_kw^{z+2}\beta_\ell \in NC_{x_t}N,
$$
to get $u=v$.

4) If $u_{t}\in NC_{x_t}N$, for some $(k,z,j)\in D_{x_t}$, $1\le \ell,t\le s$,
$u_t=\beta_\ell\beta_kw^{z+2}\beta_t$, then we set
\begin{align*}
v_{t-1}&=\beta_{j_{t-2}}w^{z_{t-1}+1}\alpha_{j_{t-1}}\beta_{j_{t-1}}\beta_\ell
=\beta{j_{t-2}}w^{z_{t-1}+2}\beta_\ell \in C_{x_t}N \\
\intertext{and}
v_t=\beta_kw^{z+2}\beta_t \in C_{x_t}N,
\end{align*}
to get $u=v$.

Now we have proved that if $V$ reliable then $\phi(L)\subseteq\xi(L)$.

\vali
Assume now that $V$ is unreliable, i.e. there is a word $x'=x_1\dots x_n$
such that $p_{x'}(j)=0$, for all $1\le j\le s$. Let
$x=bx'c=x_0x_1\dots x_n x_{n+1}\in L$. We shall first prove next claim

\vali
{\it Claim.} There is no elements $v'_i\in T_{x_i}$ for all $1\le i\le n$ such that
$w^{2n+1}=\alpha_1v'_1v'_2\dots v'_n\beta_j$.

\vali
{\it Proof of The Claim.} Assume the contrary. This means that there exists sequence
$$
y=\beta_1w^{z_1+1}\alpha_{j_1}\beta_{j_1}w^{z_2+1}\alpha_{j_2}\cdots
\beta_{j_{n-1}}w^{z_n+1}\alpha_{j_n}
$$
such that
$$
(1,,x_1,j_{1},z_1,p_1)(j_1,x_2,j_2,z_2,p_2)\cdots (j_{n-1},x_n,j_{n},z_n,p_n)
$$
is a sequence in $H$, $p_i>0$ for all $i$, and
$$
\alpha_1 y\beta_{j_n}=w^{2n+1}
$$
Now to get the number of $w$ correct on the left hand side, we must have
$$
1+(z_1+1)+1+(z_2+1)+\dots+1+(z_n+1)+1=\sum_{i=1}^n z_i +2n+1=2n+1,
$$
so $\sum_{i=1}^n z_i=0$, but this contradicts the fact that $x'$ is
critical word. This ends the proof of the claim.

\vali
Clearly $w^{2n+2}\in \phi(x)$ and we shall next show that
$w^{2n+2}\notin \xi(x)$. Assume contrary that $w^{2n+2}\in \xi(x)$,
then for all $0\le i\le n+1$ there exists $v_i\in \xi(x_i)$ such that
$w^{2n+2}=v_0\dots v_{n+1}$. Clearly the case $v_0=\alpha_1\in S$ is
only possible, since $v_0=w\beta_j\in MN$ leads to a contradiction.
Assume that $x_1=a\in \bin$ and let
$$
P=\{u\mid u\text{ is a prefix of }w^k \text{ for some integer }k\}.
$$
We divide the proof to five cases according to $v_1$:

\vali
1) If $v_1\in B$, i.e. $v_1=ww$, then $v_0v_1=\alpha_1ww\notin P$.

2) If $v_1\in NT_a$, i.e. for some $1\le \ell\le s$ and $(k,z,j)\in D_a$
$v_1=\beta_\ell\beta_kw^{z+1}\alpha_j$, then $v_0v_1\notin P$.

3) If $v_1\in C_aN$, i.e. for some $1\le \ell\le s$ and $(k,z,j)\in D_a$
$v_1=\beta_kw^{z+2}\beta_\ell$, then $v_0v_1\notin P$.

4) If $v_1\in NC_aN$, i.e. for some $1\le \ell,t\le s$ and $(k,z,j)\in D_a$
$v_1=\beta_\ell\beta_kw^{z+2}\beta_t$, then $v_0v_1\notin P$.

5) If $v_1\in T_a$, then let $t=\min\{i\mid v_i\notin T_{x_i}, 1\le i\le n\}$.
Now if $v_0v_1\dots v_{t-1}\in P$, then $v_0v_1\dots v_{t-1}=w^r\alpha_j$ for some integers $r$ and $j$,
where $1\le j\le s$.

Assume now that $t=n$. If now $v_{n+1}=w$, then $v_0v_1\dots v_nv_{n+1}\notin P$,
and if $v_{n+1}=\beta_j w\in NM$, by Claim above $v_0v_1\dots v_nv_{n+1}\ne
w^{2n+2}$ and so necessarily $t<n$.

Now we have four cases on whether $v_t\in B$, $v_t\in NT_{x_t}$,
$v_t\in C_{x_t}$ or $v_t\in NC_{x_t}N$, but like cases 1-4 above, these
cases lead to contradiction, since $v_0v_1\dots v_t\notin P$.

So we have proved that $w^{2n+2}\notin \xi(x)$ and therefore the prove
of the theorem is completed.

\end{proof}

\paragraph{Acknowledgement.} I am grateful to Juhani Karhum\"aki for guiding me to this topic originally, and especially on his support in the beginning of my career. I also sincerely thank Tero Harju for his support and especially on his comments -- comments on this manuscript, comments on my work in general and all the comments off any research topics we have had together.    


\end{document}